\documentclass[pdflatex,sn-mathphys-num]{sn-jnl}

\usepackage{graphicx}%
\usepackage{multirow}%
\usepackage{amsmath,amssymb,amsfonts}%
\usepackage{amsthm}%
\usepackage{mathrsfs}%
\usepackage[title]{appendix}%
\usepackage{xcolor}%
\usepackage{textcomp}%
\usepackage{manyfoot}%
\usepackage{booktabs}%
\usepackage{algorithm}%
\usepackage{algorithmicx}%
\usepackage{algpseudocode}%
\usepackage{listings}%
\usepackage{tikz}%
\usepackage{comment}
\usetikzlibrary{positioning}%

\newtheorem{theorem}{Theorem}
\newtheorem{lemma}{Lemma}
\newtheorem*{corollary}{Corollary}

\begin{document}

\title[Synergy as the failure of distributivity]{Synergy as the failure of distributivity}

\author*[1]{\fnm{Ivan} \sur{Sevostianov}}\email{ivan.sevostianov@weizmann.ac.il}

\author[1]{\fnm{Ofer} \sur{Feinerman}}\email{ofer.feinerman@weizmann.ac.il}

\affil[1]{\orgdiv{Department of Physics of Complex Systems}, \orgname{Weizmann Institute of Science}, \orgaddress{\street{Herzl 234}, \city{Rehovot}, \postcode{7610001}, \country{Israel}}}

\abstract{The concept of emergence, or synergy in its simplest form, is widely used but lacks a rigorous definition. Our work connects information and
set theory to uncover the mathematical nature of synergy as the failure of distributivity. It resolves the persistent self-contradiction of information decomposition theory and reinstates it as a primary route toward a rigorous definition of emergence. Our results suggest that non-distributive variants of set theory may be used to describe emergent physical systems.}


\maketitle

\section{Introduction}

Reductionism is a standard scientific approach in which a system is studied by breaking it into smaller parts. However, some of the most interesting phenomena in physics and biology appear to resist such disentanglement. In these cases, complexity emerges from intricate interactions between many predominantly simple components \cite{intro}. Such \emph{synergic} systems are typically described as ``a whole that is greater than the sum of its parts''. To pour quantitative meaning into this equation-like definition, it is natural to borrow tools from the mathematical theory that describes part-whole relationships, namely set theory. Unfortunately, for finite sets, a simple Venn diagram suffices to demonstrate that the whole ($A\cup B$) can never exceed the sum of its parts ($A$ and $B$).
\begin{eqnarray}\label{sec1:2incexcsets}
    && |A\cup B| = |A| + |B| - |A\cap B| \leq |A| + |B|
\end{eqnarray}
In fact, the trivial interaction, $A\cap B$, between the two parts of the system decreases the size of the whole rather than increasing it. 

To allow for more intricate interactions, one can turn to the realm of random variables. It is well known that measuring the outcome of two random variables can provide more information than the sum of what is obtained when measuring each separately. Moreover, the textbook description of the interactions between random variables often involves set-theoretical-like Venn diagrams \cite{coverthomas}. These two facts lead to the intriguing possibility that random variables may lend themselves to a mathematical description of non-trivial whole-part relationships.
 
 Take two discrete variables $W$ and $Z$: the information $W$ contains about $Z$ is determined by the mutual information function $I(W;Z)$ \cite{Shannon}. Cases for which $W$ can be presented as a joint distribution $W=(X,Y)$ allow us to compare the whole against its parts:
\begin{eqnarray}\label{intro:wholevspartsinfo}
    && I((X, Y); Z) \lesseqqgtr I(X; Z) + I(Y; Z)
\end{eqnarray}
In other words, looking at both system parts together can convey either more or less information than their added values. Therefore, and in contrast to eq. \ref{sec1:2incexcsets}, this formalism can be used to describe synergy.

\begin{figure}
    \begin{center}  
        \includegraphics[scale=0.42]{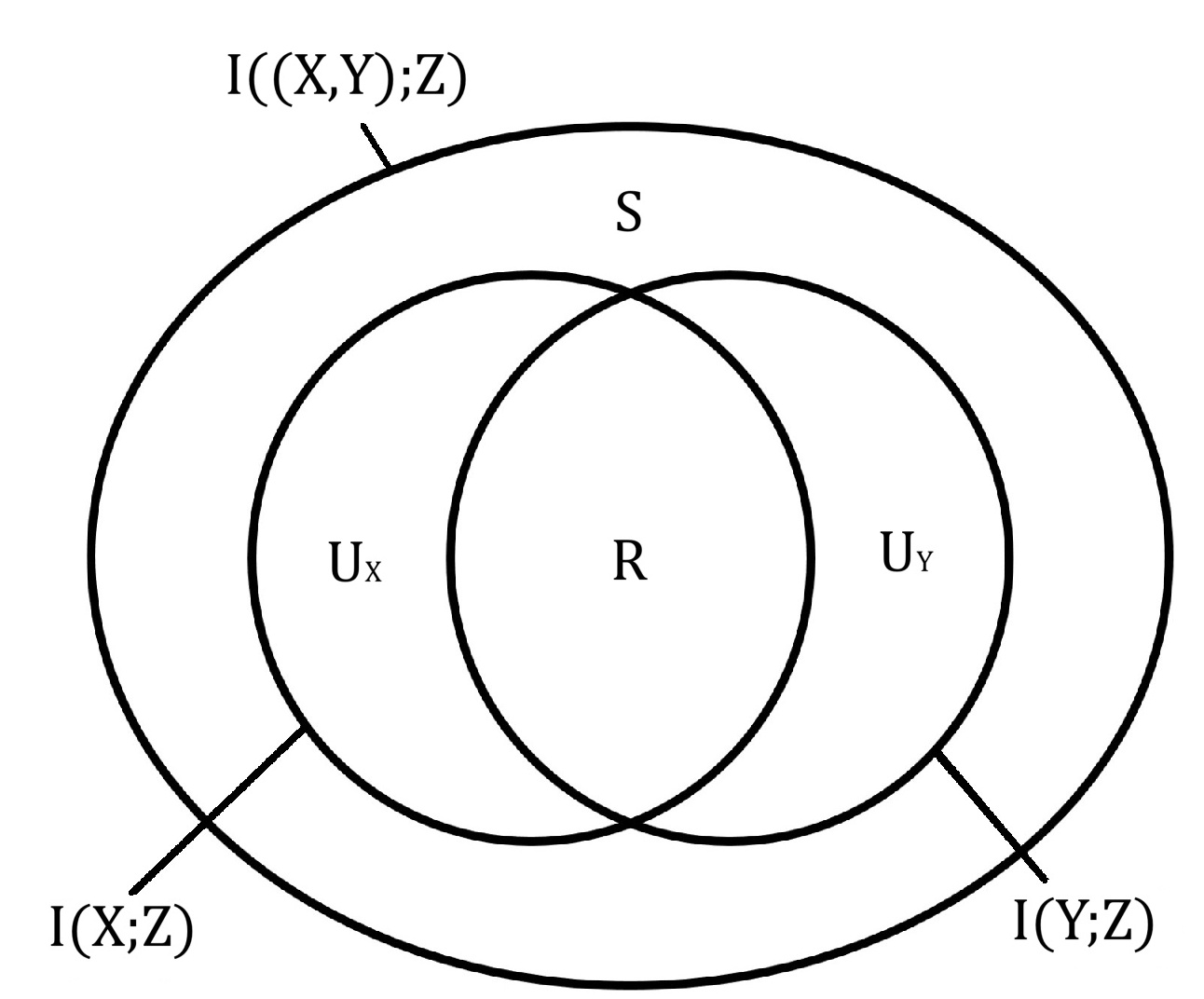}
        \caption{A diagram representing different types of mutual information in the system of two random variables $X, Y$ about $Z$ by the part-whole relations. Redundant information is shared between $X$ and $Y$, unique information is a part of just one of them, while synergistic information is something that is only contained in the joint distribution, but not individual sources on their own.}
        \label{fig:bivariatePID}
    \end{center}
\end{figure}

In their seminal paper \cite{WilliamsAndBeer}, Williams and Beer proposed the framework of \emph{partial information decomposition} as a way of assessing the underlying structure of a two discrete random variable system and quantifying the amount of synergy between its parts. They suggested that, much like a set of elements, each variable can be decomposed into separate information ``subsets''. These \emph{information atoms} are assumed to have non-negative size and can either be shared between variables ($R$) or uniquely present in only one of them ($U_X, U_Y$) (see Fig. \ref{fig:bivariatePID})
\begin{eqnarray}\label{intro:PIDequations}
	&& I(X; Z) = R + U_X, \cr
	&& I(Y; Z) = R + U_Y, \cr
	&& I((X, Y); Z) = R + U_X + U_Y + S, \cr 
        && R, U_X, U_Y, S \geq 0
\end{eqnarray}
An additional synergy term ($S$)  was artificially introduced to provide a simple mechanism that allows the whole to be greater than the sum of its parts
\begin{eqnarray}
    && I((X, Y); Z) - I(X; Z) - I(Y; Z) = S - R > 0 \text{ iff } S > R
\end{eqnarray}

A series of papers \cite{BROJA,Artemy,Identity,ZeroError2} focused on calculating these atoms' sizes by fixing the single remaining degree of freedom in equations (\ref{intro:PIDequations}). No consensus has yet been reached regarding a single physical solution. Meanwhile, the field of applications is getting wider \cite{GreaterThanTheParts, Lizier}. Recent works extend the theory to continuous variables \cite{PIDextension, continuous2}, introduce causality \cite{CausalEmergence, Tononi}, and consider quantum information \cite{QuantumPID}.

Unfortunately, partial information decomposition has a significant drawback that puts the whole approach into question: no extension beyond two variables is possible without a fundamental self-contradiction \cite{BROJ}. Some authors attempted to resolve this by abandoning the basic properties required of information atoms, including their non-negativity \cite{Pointwise, IncePHD}.

In what follows, we reconsider the foundations of partial information decomposition and pinpoint the source of its long-standing self-contradictions. To do this, we follow  Hu \cite{Hu, HuTheorem, Yeung} to establish a rigorous relation between information and set theories and highlight a fundamental distinction between them: random variables, unlike sets, do not adhere to the union/intersection distributivity axiom \cite{TerryTao}. This leads us to study a distributivity-free variant of set theory as a possible self-consistent theory of information atoms. Within this framework, we demonstrate that the presence of synergistic properties is a direct consequence of the broken axiom. In the case of $N=3$ random variables, we show that the amount of synergistic information precisely coincides with the extent to which distributivity is breached. The acquired understanding allows us to resolve the contradictions and suggest a coherent multivariate theory, which may provide the foundations for quantifying emergence in large systems.

\section{Results}

\subsection{Set-theoretic approach to information}

In this section, we formalize the distinction between finite sets and discrete random variables. Clearly, it is linked to the synergistic behavior of the latter. We will first focus on a special illustrative example: the XOR gate. This system contains neither redundant nor unique information, which will emphasize the peculiar properties of synergy. A more general discussion, including arbitrary random variables, will be presented in the next section.

\paragraph{Basic random variable operations}

Some set-theoretic operations have straightforward extensions to random variables. The first of these emerges by comparing eqs. \ref{sec1:2incexcsets} and \ref{intro:wholevspartsinfo} and relates the joint distribution of two variables with the union operator ($\cup$). One can now go on to define random variable inclusion as
\begin{eqnarray}\label{sec1:inclusion}
    && X \subseteq Y \Leftrightarrow \exists Z : X\cup Z = Y 
\end{eqnarray}
which is, actually, equivalent to $X$ being a deterministic function of $Y$. 

The inclusion-exclusion formula for two random variables reveals mutual information as the size of the intersection between two random variables
\begin{eqnarray}\label{sec1:incexc2}
    && H(X\cup Y) = H(X, Y) = H(X) + H(Y) - I(X; Y),
\end{eqnarray} 
where Shannon entropy $H$ is regarded as a measure on the random variable space. Indeed, it complies with the many properties required of a mathematical measure \cite{Measure}: non-negativity, monotonicity and subadditivity. Further, entropy is zero only for deterministic functions, which play the role of empty set
\begin{eqnarray}
    && H(X) \geq 0, \cr 
    && X \subseteq Y \Rightarrow H(X) \leq H(Y), \cr 
    && H(X_1 \cup X_2 \cup \dots \cup X_N) \leq H(X_1) + H(X_2) + \dots + H(X_N), \cr
    && H(X) = 0 \Leftrightarrow X = \emptyset
\end{eqnarray}
A rigorous definition of intersection ($\cap$) needs to comply with the inclusion order (\ref{sec1:inclusion}): $X\cap Y \subseteq X, X\cap Y \subseteq Y$. Unfortunately, a random variable satisfying both conditions does not always exist \cite{ZeroError1}. Nonetheless, a physically sensible intersection may be inferred in several cases:
\begin{eqnarray}\label{sec1:intersection}
    && H(X \cup Y) = H(X) + H(Y) \Leftrightarrow X\cap Y = \emptyset, \cr 
    && X \subseteq Y \Leftrightarrow H(X) = I(X; Y) \Leftrightarrow X\cap Y = X
\end{eqnarray}
These simple parallels between information theory and set theory are enough to study information decomposition in a random variable XOR gate. 

\paragraph{Subdistributivity}

Our set-theoretic intuition for random variables breaks down even further when considering the XOR gate: three pairwise independent fair coins $O_1, O_2, O_3$ with an additionally imposed higher order interaction - parity rule $O_3 = O_1 \oplus O_2$
\begin{center}
\begin{tabular}{|c|c|c|c|}
     \hline
       Probability & $O_1$ & $O_2$ & $O_3$ \\ [0.5ex] 
     \hline
     $1/4$ & 0 & 0 & 0  \\ 
     \hline
     $0$ & 0 & 0 & 1 \\
     \hline
     $0$ & 0 & 1 & 0 \\ 
     \hline
     $1/4$ & 0 & 1 & 1 \\ 
     \hline
     $0$ & 1 & 0 & 0 \\ 
     \hline
     $1/4$ & 1 & 0 & 1 \\
     \hline
     $1/4$ & 1 & 1 & 0 \\ 
     \hline
     $0$ & 1 & 1 & 1 \\ 
     \hline
    \end{tabular}
\end{center}
The pairwise independence dictates $O_2 \cap O_3 = O_1 \cap O_3 = \emptyset$, while the parity rule makes $O_3$ a deterministic function of the joint distribution $(O_1, O_2)$
\begin{eqnarray}\label{sec1:XORsubset}
    && O_3 \subseteq (O_1\cup O_2) \Rightarrow (O_1\cup O_2) \cap O_3 = O_3
\end{eqnarray}
A simple conclusion from this facts is that the XOR-gate variables do not comply with the set-theoretic axiom of distributivity
\begin{eqnarray}
    && (O_1\cup O_2) \cap O_3 = O_3 \neq \emptyset = (O_1\cap O_3) \cup (O_2\cap O_3)
\end{eqnarray}
Nevertheless, it can be shown that a weaker relation of \emph{subdistributivity} holds for any three random variables (SI, Lemma \ref{SI:lemma:subdistributivity})
\begin{eqnarray}\label{subdistributivityaxiom}
    && (X \cup Y) \cap Z \supset (X \cap Z) \cup (Y \cap Z)
\end{eqnarray}
Even though it is evident that random variables are quite different from sets, we argue that some of the logic behind partial information decomposition may be recovered by extending set-theoretic notions, such as the inclusion-exclusion principle and Venn diagrams, to non-distributive systems. 

\paragraph{Inclusion-exclusion formulas}

The inclusion-exclusion formula for the XOR gate can be obtained by repeatedly applying the 2-variable equation (\ref{sec1:incexc2}) and using that $I(X;Y) = H(X \cap Y)$ when the intersection exists
\begin{eqnarray}\label{sec1:XORincexc}
    && H(O_1 \cup O_2 \cup O_3) = \cr 
    && = H(O_1 \cup O_2) + H(O_3) - H((O_1\cup O_2)\cap O_3) = \cr 
    && = H(O_1) + H(O_2) + H(O_3) - H((O_1\cup O_2)\cap O_3)
\end{eqnarray}
It disagrees with the analogous set-theoretic formula (for non-intersecting sets) only in the last term, which is non-zero precisely due to the subdistributivity. Note that while the rest of the terms are symmetric with respect to the permutation of indices, expression $(O_1\cup O_2)\cap O_3$ is not as it explicitly depends on the order of derivation. This essentially leads to three different inclusion-exclusion formulas. Nonetheless the size of the distributivity-breaking term remains invariant 
\begin{eqnarray}
    && H((O_1\cup O_2)\cap O_3) = H((O_1\cup O_3)\cap O_2) = H((O_2\cup O_3)\cap O_1)
\end{eqnarray}

\paragraph{Construction of Venn-type diagram for the XOR gate}

The non-uniqueness of inclusion-exclusion formulas complicates the construction of Venn diagrams. A way of tackling this as well as some further intuition can be traced via our XOR gate example. 

In set theory, Venn diagrams act as graphical representations of the inclusion-exclusion principle. The inclusion-exclusion formula computes the size of union as a sum of all possible intersections between the participating sets. For correct bookkeeping, this is done with alternating signs that account for the \emph{covering number} - the number of times each intersection is counted as a part of some set. In classical set theory, the covering number of an intersection is trivially the number of sets which are being intersected. However, (\ref{sec1:XORincexc}) includes the distributivity-breaking term which is absent from this classical theory and whose covering number is not evident. It appears with a negative sign which signifies an even-times covered region. In this three variable system, the only even alternative is a $2$-covered region. From another perspective, in each of the three possible formulas $O_k$ is covered once by itself and one more time by the union $O_i\cup O_j$ (though not by $O_i$ or $O_j$ individually). As for the size of this region, independent of $k$, it measures at $1$ bit of information. Denoting this area as $\Pi_s$, we have
\begin{eqnarray}\label{sec1:XORsatomdef}
    && \Pi_s[2] = H((O_i\cup O_j)\cap O_k) = H(O_k) = 1 \text{ bit},
\end{eqnarray}
where the covering number is indicated in the brackets $[]$. The region $\Pi_s$ does not contain enough information to describe the whole system, as the total amount of entropy in the XOR gate is $H(O_1\cup O_2\cup O_3) = 2 \text{ bit}$. To draw the diagram, we need to find the covering of the remaining $2-1=1$ bit region/regions. This may be accomplished by borrowing two properties of set-theoretic diagrams. 

First of all, in a system of $N$ arbitrary random variables $X_1, \dots X_N$ the total entropy of the system is equal to the sum of all diagram regions $\Pi_i[c_i]$
\begin{eqnarray}\label{totalarea}
    && H(X_1, \dots X_N) = \sum_i \Pi_i[c_i]
\end{eqnarray}

Second, the sum of individual variables' entropies is equal to the sum of region sizes times their corresponding covering numbers $c_i$ 
\begin{eqnarray}\label{conservationlaw}
    && H(X_1) + H(X_2) + \dots + H(X_N) = \sum_i c_i\Pi_i[c_i]
\end{eqnarray}
These properties may be viewed as the \emph{information conservation law}: adding new sources should either introduce new information or increase the covering of existing regions. Information cannot spontaneously arise nor vanish.

Let us assume that in addition to $\Pi_s$ the diagram of the XOR gate contains several more regions $\Pi_{j\neq s}$. To calculate their sizes and coverings we apply (\ref{totalarea}-\ref{conservationlaw}) 
\begin{eqnarray}
    && \sum_{j\neq s} (c_j - 1)\Pi_j[c_j] = 0 \text{ bit}
\end{eqnarray}
We use the fact that information is non-negative and discard meaningless empty regions. The above equation then allows for a single $1$-bit region, which is covered once 
\begin{eqnarray}
    && \Pi_g[1] = 1 \text{ bit}
\end{eqnarray}
To respect the physical meaning behind the diagram regions as pieces of information, we demand the structure of the diagram to be well-defined. In other words, despite the existence of three different versions of inclusion-exclusion formula (\ref{sec1:XORincexc}), they are all assumed to describe the \emph{same} system. Indeed, our result remains invariant with respect to index permutations in terms of region sizes and covering numbers. 

In regard to the shape of the Venn diagram, this assumption dictates along with (\ref{sec1:XORsubset},\ref{sec1:XORsatomdef}) that the region $\Pi_s$ corresponds to all variables at the same time
\begin{eqnarray}
    && \Pi_s = H(O_1) = H(O_2) = H(O_3)
\end{eqnarray}
One can think of $\Pi_s$ as a $2$-covered triple intersection between $O_1, O_2$ and $O_3$. This is a drastic divergence from classical set theory where an intersection between $n$ sets is covered exactly $n$ times. As we shall see, without distributivity, $n$ variables can have multiple intersection regions with different covering numbers $1 \leq c\leq n$.

Moving on to the second region in this system: $\Pi_g$ appears as a leftover when taking the difference between the whole system and $\Pi_s$ and by set-theoretic intuition it does not intersect with $O_k$ for any $k$. As such it is not a part of any single variable. 

Finally, we combine all findings into a system of equations, which generates the Venn-type diagram of the information distribution inside XOR gate (Fig. \ref{fig:XORdiagram})
\begin{eqnarray}\label{sec1:XORdiagrameqs}
    && H(O_1) = \Pi_s, \cr 
    && H(O_2) = \Pi_s, \cr 
    && H(O_3) = \Pi_s, \cr 
    && H(O_1 \cup O_2) = \Pi_s + \Pi_g, \cr
    && H(O_2 \cup O_3) = \Pi_s + \Pi_g, \cr
    && H(O_1 \cup O_3) = \Pi_s + \Pi_g, \cr
    && H(O_1\cup O_2\cup O_3) = \Pi_s + \Pi_g
\end{eqnarray}
In usual Venn diagrams, intersections represent correlations between the different parts. Similarly, in the XOR gate the higher-order parity interaction added on top of the non-correlated variables is responsible for the appearance of the $2$-covered triple intersection.
\begin{figure}
    \begin{center}  
        \includegraphics[scale=0.42]{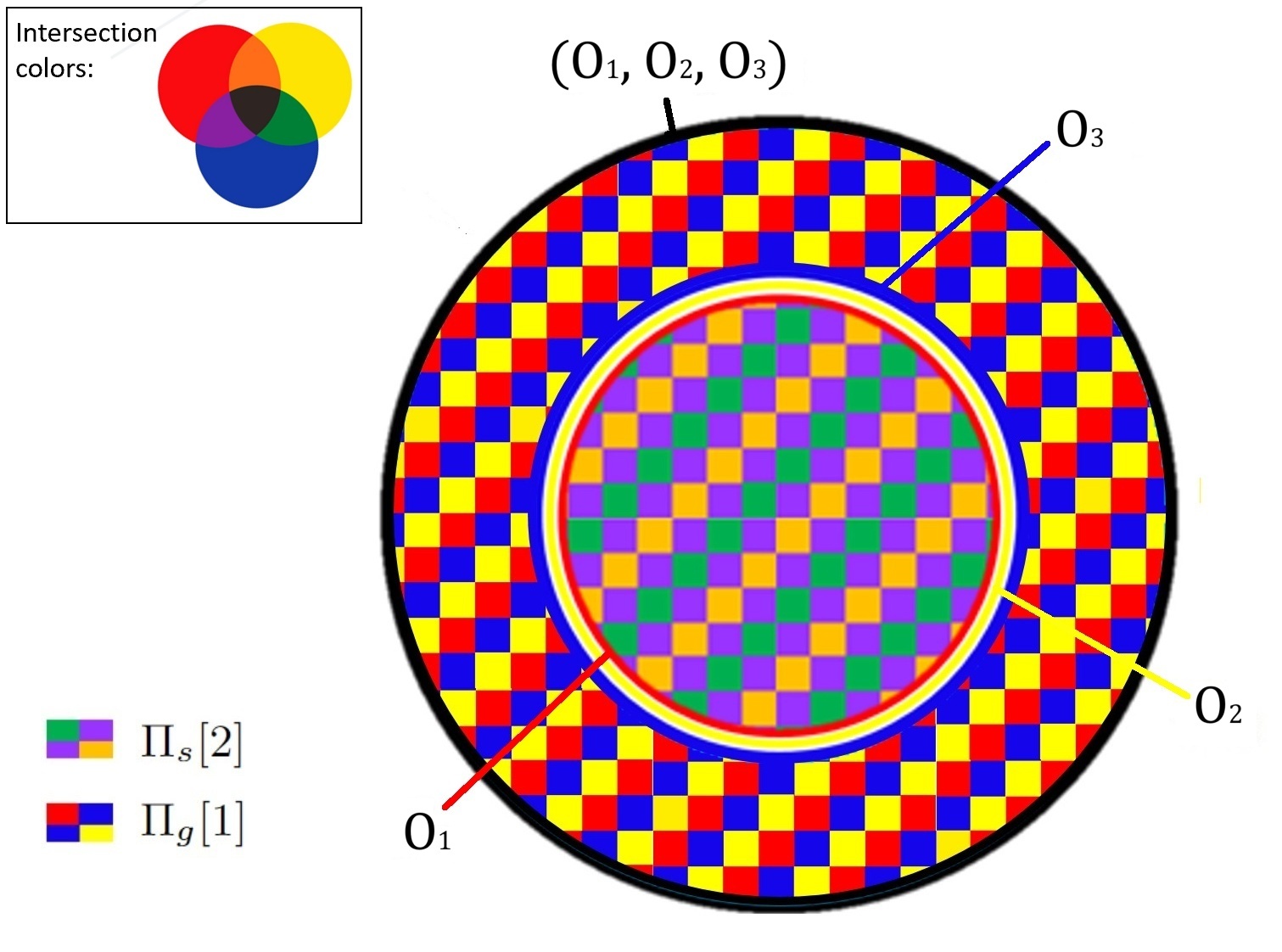}
        \caption{A Venn-type diagram for the XOR gate. Each variable is represented by a primary color circle (red, yellow, blue) while the outer circle outlines the whole system. Of the total $2$ bits of the XOR gate, one is covered two times and is represented by the inner disk. Since it is covered twice this area is colored by  pairwise color-blends (orange, purple, and green). Since it is covered by three variables it includes patches of all three possible blends. A critical difference between this diagram and a set-theoretic one is that even though the three variables have no pairwise intersections, the inner disk representing the 'mutual' content of all three variables is non-empty. The remaining $1$ bit is covered once and resides only inside the joint distribution. Since this area is covered once, it is colored by primary colors.  Patches of all three colors are used since this area does not belong to any single variable.} 
        \label{fig:XORdiagram}
    \end{center}
\end{figure}

\paragraph{Synergy as an information atom}

We can compare our set-theory inspired results against the expectations of partial information decomposition. Namely, equations (\ref{intro:PIDequations}) state that the information $O_1$ and $O_2$ carry about $O_3$ can be described by the atoms $R = U_1 = U_2 = 0 \text{ bit}, S = 1 \text{ bit}$. The left side of each line in (\ref{intro:PIDequations}) may be rewritten by definition as an intersection of random variables
\begin{eqnarray}
    && I(X; Z) = H(X \cap Z), \cr 
    && I(Y; Z) = H(Y \cap Z), \cr 
    && I((X, Y); Z) = H((X\cup Y) \cap Z)
\end{eqnarray}
For XOR gate the former two are empty, while the last line links the original definition of synergistic information to the non-set-theoretic term of the inclusion-exclusion formula (\ref{sec1:XORincexc}) and the peculiar region of the corresponding diagram
\begin{eqnarray}\label{sec1:synergy}
    && S = I((O_1, O_2); O_3) = H((O_1\cup O_2) \cap O_3) = \Pi_s
\end{eqnarray}
Curiously, synergistic behavior of mutual information does not contradict the subadditivity of entropy. The synergistic information piece $S$ is not new to the system and is always contained in the variables' full entropy.

The nature of \emph{ghost atom} $G = \Pi_g$ is deeply connected to this outcome even though it does not explicitly participate in the decomposition. Consider the individual contributions by each of the sources $O_1, O_2$ 
\begin{eqnarray}
    && I(O_{i=1, 2}; O_3) = H(O_i) + H(O_3) - H(O_i\cup O_3)
\end{eqnarray}
Using (\ref{sec1:XORdiagrameqs}) we can rewrite this in terms of information atoms 
\begin{eqnarray}\label{sec1:ghostvssynergy}
    && I(O_i; O_3) = \Pi_s + \Pi_s - (\Pi_s + \Pi_g) = S - G = 0
\end{eqnarray}
The equality between the synergistic and ghost atoms ensures that the former is exactly cancelled from the individual contribution by each source. Synergistic information is, of course, still present in the 'whole' (\ref{sec1:synergy}). This circumstance is responsible for creating the illusion of synergy appearing out of nowhere when sources are combined.

\subsection{General trivariate decomposition}\label{sec2}

The XOR gate example studied above is a degenerate example with a sole synergistic information atom. We will now expand our approach into a system with non-synergistic components to provide a complete description of three arbitrary variables.

\paragraph{Extended random variable space}

\begin{figure}
    \begin{center}  
        \includegraphics[scale=0.42]{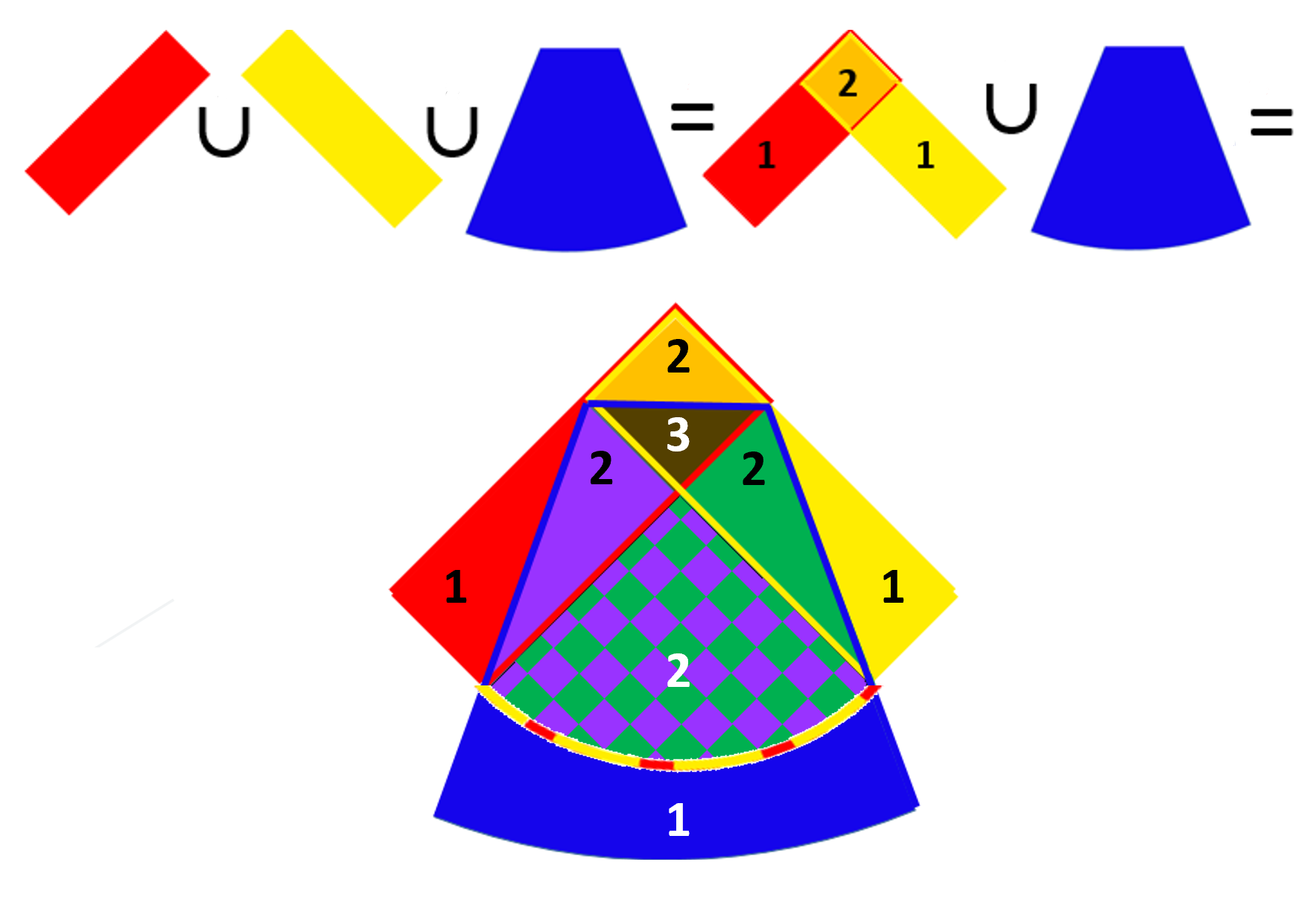}
        \caption{A single realization of inclusion-exclusion principle for three variables. The new region, corresponding to the distributivity-breaking difference is represented via a checkered pattern. Covering numbers are written for each sector and highlighted by the colors. This is not a full Venn-type diagram that defines the information atoms, hence its structure is clearly not invariant with respect to variable permutations.}
        \label{fig:3varInfoSpace}
    \end{center}
\end{figure}

The lack of a proper description for information intersections severely limits our ability to decompose the information content of more general random variable systems. Our solution for this issue is inspired by an elegant duality between set theory and information quantities found by Hu in \cite{Hu} and further elaborated in \cite{HuTheorem}. It simply extends the space of random variables to include all elements produced by operations $\cup, \cap, \backslash$ (\ref{intro:wholevspartsinfo}, \ref{sec1:intersection}, \ref{sec2:differencedef}). Entropy is extended as a (non-negative) measure $\hat H$ such that
\begin{eqnarray}\label{sec2:measure}
    && \hat H(X) = 0 \Leftrightarrow X = \emptyset, \cr 
    && X\cap Y = \emptyset \Leftrightarrow \hat H(X\cup Y) = \hat H(X) + \hat H(Y)
\end{eqnarray}

To approach the problem of characterizing information atoms in the trivariate case, we derive the corresponding inclusion-exclusion formula. As stated previously, the bivariate version (\ref{sec1:incexc2}) holds without alterations (SI, Lemma \ref{SI:lemma:inclusionexclusion2}). Now, in contrast, we get a distributivity-breaking difference term, which, to make matters even worse, depends on the order of derivation (SI, Theorem \ref{SI:theorem:incexc3}). A visualisation for one possible variant of this formula can be seen in Fig. \ref{fig:3varInfoSpace}
\begin{eqnarray}\label{incexc3}
    && \hat H(X_1\cup X_2\cup X_3) = \cr 
    && = \hat H(X_1) + \hat H(X_2) + \hat H(X_3) - \cr 
    && - \hat H(X_1\cap X_2) - \hat H(X_1\cap X_3) - \hat H(X_2\cap X_3) + \cr 
    && + \hat H(X_1\cap X_2\cap X_3) - \Delta \hat H,
\end{eqnarray}  
where for any permutation $\sigma$
\begin{eqnarray}
    && \Delta \hat H = \hat H(((X_{\sigma(1)}\cup X_{\sigma(2)})\cap X_{\sigma(3)}) \backslash ((X_{\sigma(1)}\cap X_{\sigma(3)})\cup(X_{\sigma(2)}\cap X_{\sigma(3)})))
\end{eqnarray}
and the difference is defined as 
\begin{eqnarray}\label{sec2:differencedef}
    && D = X\backslash Y \Leftrightarrow D \cap Y = \emptyset, D \cup (X\cap Y) = X
\end{eqnarray}
In general, due to subdistributivity the difference may not be unique (SI, \ref{differenceexample}). Its size, on the other hand, is fixed as $\hat H(X\backslash Y) = \hat H(X) - \hat H(X\cap Y)$.

\paragraph{Set-theoretic solution}

Before going to arbitrary variables, consider a system where distributivity axiom holds. Under such condition the setup becomes effectively equivalent to set theory. A trivariate system can therefore be illustrated by the same Venn diagram as that of $3$ sets 
\begin{eqnarray}\label{sec2:settheoreticsolution}
    && H(X_1) = \Pi_{\{1\}} + \Pi_{\{1\}\{2\}} + \Pi_{\{1\}\{3\}} + \Pi_{\{1\}\{2\}\{3\}}, \cr 
    && H(X_2) = \Pi_{\{2\}} + \Pi_{\{1\}\{2\}} + \Pi_{\{2\}\{3\}} + \Pi_{\{1\}\{2\}\{3\}}, \cr 
    && H(X_3) = \Pi_{\{3\}} + \Pi_{\{1\}\{3\}} + \Pi_{\{2\}\{3\}} + \Pi_{\{1\}\{2\}\{3\}}, \cr 
    && H(X_1, X_2) = \Pi_{\{1\}} + \Pi_{\{2\}} + \Pi_{\{1\}\{2\}} + \Pi_{\{1\}\{3\}} + \Pi_{\{2\}\{3\}} +  \Pi_{\{1\}\{2\}\{3\}}, \cr 
    && H(X_1, X_3) = \Pi_{\{1\}} + \Pi_{\{3\}} + \Pi_{\{1\}\{2\}} + \Pi_{\{1\}\{3\}} + \Pi_{\{2\}\{3\}} +  \Pi_{\{1\}\{2\}\{3\}}, \cr 
    && H(X_2, X_3) = \Pi_{\{2\}} + \Pi_{\{3\}} + \Pi_{\{1\}\{2\}} + \Pi_{\{1\}\{3\}} + \Pi_{\{2\}\{3\}} +  \Pi_{\{1\}\{2\}\{3\}}, \cr 
    && H(X_1, X_2, X_3) = \cr 
    && = \Pi_{\{1\}} + \Pi_{\{2\}} + \Pi_{\{3\}} + \Pi_{\{1\}\{2\}} + \Pi_{\{1\}\{3\}} + \Pi_{\{2\}\{3\}} +  \Pi_{\{1\}\{2\}\{3\}}
\end{eqnarray}
By calculating the sizes of atoms, we derive (SI, \ref{SI:settheor3derivation}) the criterion for their non-negativity: the whole must be less or equal to the sum of the parts 
\begin{eqnarray}\label{sec2:settheoreticsolutioncond}
    && I(X_1, X_2; X_3) - I(X_1; X_3) - I(X_2; X_3) \leq 0
\end{eqnarray}

\paragraph{Main result: arbitrary trivariate system}

\begin{figure}
    \begin{center}  
        \includegraphics[scale=0.42]{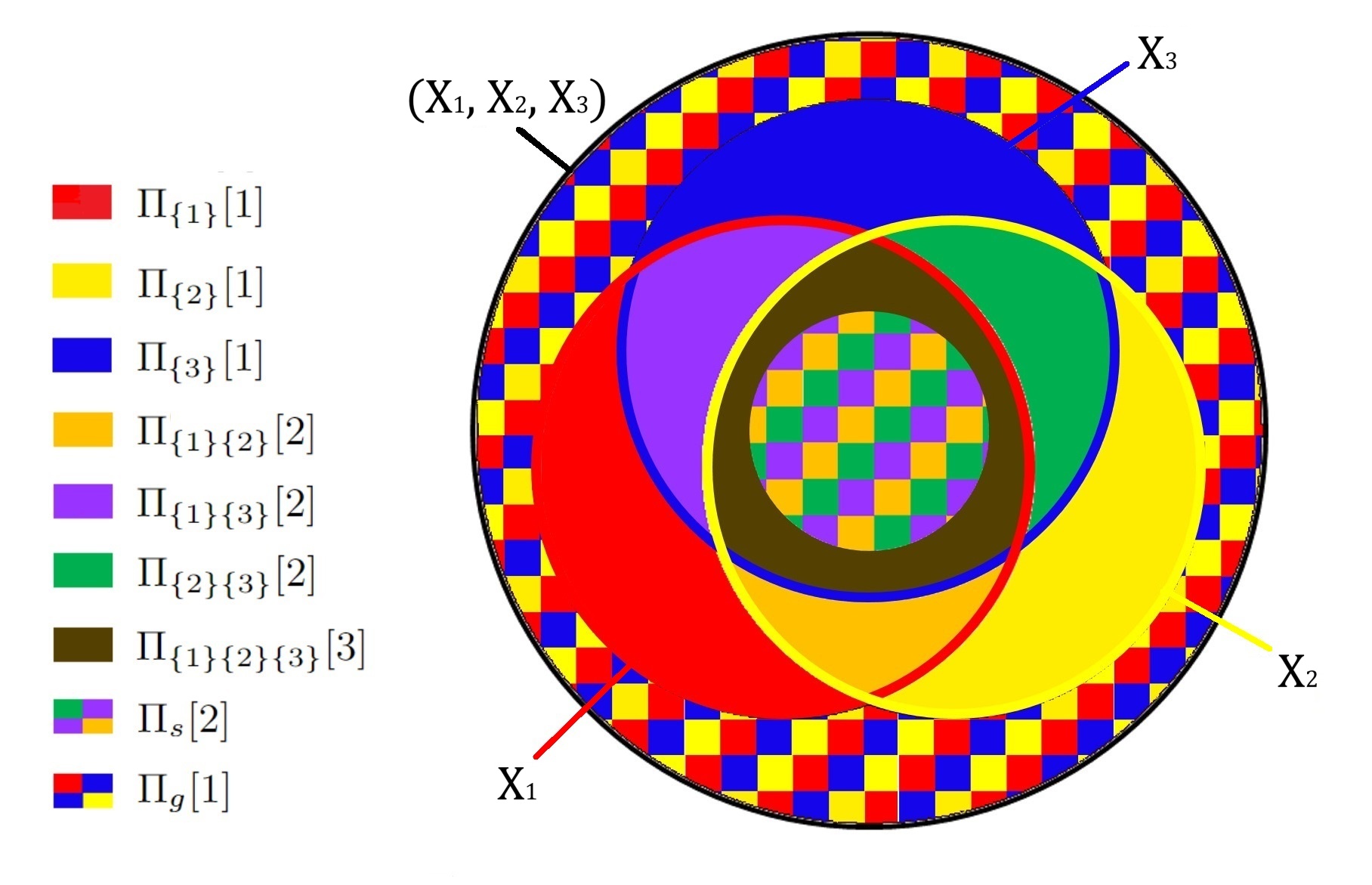}
        \caption{A graphical illustration for the general solution of the trivariate problem. Compared to the Venn diagram for $3$ sets, two new regions here are the $2$-covered part of triple intersection $\Pi_s$ (synergistic atom) and a ghost atom $\Pi_g$, which is not a part of any single initial variable. Similarly to Fig. \ref{fig:3varInfoSpace}, colors indicate the coverings: 3 primary colors (red, yellow, blue or their checkered combination) correspond to 1-covered atoms, the overlay of any 2 colors (orange, purple, green or their checkered combination) is 2-covered and the overlay of all 3 colors (brown) is 3-covered.}
        \label{fig:3PHD}
    \end{center}
\end{figure}

At this point, we have studied two opposite cases: a completely synergistic system (XOR gate) and one without any synergy (set-theoretic solution). To describe three arbitrary variables, any general decomposition must be able to replicate both of them. It turns out that a combination of the already known atoms (Fig. \ref{fig:3PHD}) suffices in providing a non-negative decomposition (presented in detail in SI, \ref{3PHDequations}; for proof see SI, Lemma \ref{SI:lemma:arbitrary3PHD})
\begin{eqnarray}
    && H(X_i) = \Pi_s + \sum_{\text{set-theor. atoms}}\Pi, \cr 
    && H(X_i, X_{j\neq i}) = \Pi_s + \Pi_g + \sum_{\text{s.t. atoms}} \Pi, \cr 
    && H(X_1, X_2, X_3) = \Pi_s + \Pi_g + \sum_{\text{s.t. atoms}} \Pi, \cr 
    && \Pi_s = \Pi_g
\end{eqnarray}
This is the minimal solution to the problem as it contains the smallest set of necessary atoms. The whole and parts are now related by the difference of two terms
\begin{eqnarray}\label{sec2:interinfoatoms}
    && I((X_1, X_2); X_3) - I(X_1; X_3) - I(X_2; X_3) = \Pi_s - \Pi_{\{1\}\{2\}\{3\}} \lesseqqgtr 0
\end{eqnarray}  
We can gain major insight by substituting the left side using the inclusion-exclusion formulas (\ref{sec1:incexc2}, \ref{incexc3}). Remember that the only $3$-covered area in the system is $X_1\cap X_2\cap X_3$. Therefore, the size of $\Pi_s$ is determined by the distributivity-breaking difference
\begin{eqnarray}
    && I((X_1, X_2); X_3) - I(X_1; X_3) - I(X_2; X_3) = \Delta \hat H - \hat H(X_1\cap X_2\cap X_3), \cr
    && \Pi_{\{1\}\{2\}\{3\}} = \hat H(X_1\cap X_2\cap X_3), \cr 
    && \Pi_s = \Delta \hat H
\end{eqnarray}
To find the physical meaning behind the recovered solution, we once again compare it to the partial information decomposition of the same system. Only four of the diagram regions (Fig. \ref{fig:3PHD}) appear in the corresponding equations
\begin{eqnarray}\label{3PIDequations}
    && I(X_1; X_3) = \Pi_{\{1\}\{2\}\{3\}} + \Pi_{\{1\}\{3\}}, \cr 
    && I(X_2; X_3) = \Pi_{\{1\}\{2\}\{3\}} + \Pi_{\{2\}\{3\}}, \cr 
    && I((X_1, X_2); X_3) = \Pi_{\{1\}\{2\}\{3\}} + \Pi_s + \Pi_{\{1\}\{3\}} + \Pi_{\{2\}\{3\}}
\end{eqnarray}
The result fully captures the structure behind Williams and Beer's definitions
\begin{eqnarray}
    && \Pi_{\{1\}\{2\}\{3\}} \equiv \text{Redundancy}, \cr 
    && \Pi_{\{1\}\{3\}} \equiv \text{Unique information in } X_1, \cr 
    && \Pi_{\{2\}\{3\}} \equiv \text{Unique information in } X_2, \cr 
    && \Pi_s \equiv \text{Synergy}
\end{eqnarray}

We have thus shown how information synergy naturally follows from set-theoretic arguments. The synergistic contribution is contained in the entropy of the parts and is precisely equal to the distributivity-breaking difference $\Delta \hat H$. The interaction responsible for the synergistic contribution is depicted on the Venn diagram as an intersection with unconventional covering number $\Pi_s$. Finally, the illusion of a whole being greater than the sum of its parts comes from the fact that the mutual information terms on the left-hand size of eqs. \ref{intro:PIDequations} do not account for all regions of the Venn-diagram (Fig. \ref{fig:3PHD}).

\subsection{Towards a multivariate information decomposition}

In this section we lay the foundation for a consistent theory of multivariate decomposition and resolve the contradictions between partial information decomposition axioms \cite{BROJ}.

\paragraph{Information atoms based on part-whole relations}

To rigorously define the information atoms, we may think of them as basic pieces of information, which make up all more complex quantities. Previously we have used the inclusion-exclusion principle to break down the entropy of the whole system into smaller parts step by step. Even without writing the formula for $N$ variables, one can find the general form of the terms participating in this process
\begin{eqnarray}\label{sec3:term}
    && \Xi[C] = \bigcap \left(\bigcup X_{i}\right)
\end{eqnarray}
The covering number $C$ is defined trivially as the number of intersecting union-brackets in (\ref{sec3:term}) and determines the sign of the associated term by the inclusion-exclusion principle. Similarly to the Möbius inversion used in set theory \cite{Mobius}, the decomposition of non-distributive space will rely on the inclusion order lattice $(L^{\Xi}, \subseteq)$ of terms $\Xi$. A general description of the decomposition through part-whole relations was proposed in \cite{BitsAndPieces} in the form of the \emph{parthood table}. It is a matrix with entries $0$ or $1$, which define whether a given atom $\Pi$ is a part of a particular larger information piece in the form of inclusion-exclusion terms (\ref{sec3:term}) 
\begin{eqnarray}\label{sec3:parthood}
    && \hat H(\Xi_i[C_i]) = \sum_j f_{ij}\Pi_j[c_j], \cr 
    && f_{ij} = 0, 1
\end{eqnarray}
The parthood table depends on the initial variables through the \emph{monotonicity} axiom, i.e. compliance with the inclusion lattice $(L^{\Xi}, \subseteq)$
\begin{eqnarray}\label{sec3:monotonicity}
    && \Xi_i \subseteq \Xi_j \Rightarrow\forall k \text{ } f_{ik} \leq f_{jk}   
\end{eqnarray}
It relates the table's entries within themselves by a simple rule: if one $\Xi$ term is included in the other, all the atoms from the decomposition of former should be present in the decomposition of the latter.

The summands $\Pi$ are non-negative functions and represent the sizes of atoms. The covering number $c_j$ of each atom is defined by the coverings of inclusion-exclusion terms $C_i$
\begin{eqnarray}\label{sec3:coverings}
    && c_j = \max_{i: f_{ij} = 1} C_i
\end{eqnarray}
This rule remains unchanged from the classical set theory.

The information conservation law (\ref{conservationlaw}) is the final condition that preserves the physical meaning of the covering numbers - the number of times the same information appears in the system.

The existence of a general solution for $N$ variables is not guaranteed. Besides, linear system (\ref{sec3:parthood}) is undetermined for $N>2$. For a specific set of degenerate cases it is, however, still possible to calculate the sizes of all atoms. We will next list several such examples while specifying how information is distributed among their different parts:

\paragraph{Set-theoretic solution for $N$ variables}

In a distributive system the solution is a particular case of Möbius inversion \cite{Mobius} (SI, \ref{SI:settheory1}-\ref{settheoreticatomscond}). Mutual information as a function of random variables becomes subadditive (SI, Lemma \ref{SI:lemma:nosynergyinsettheory}) proving that the lack of distributivity is a necessary condition for emergence.

\paragraph{XOR gate}

The solution found for XOR gate is unique in the parthood table formalism (SI, Theorem \ref{SI:theorem:XORunique}). This reinforces our proposal of synergistic and ghost atoms as physical entities.

\paragraph{$N$-parity}

Generalizing XOR gate to an arbitrary number of variables yields the $N$-parity setup. It allows a solution of the similar form (SI, \ref{Nparity}-\ref{Nparity2}) 
\begin{eqnarray}
    && \Pi_s[2] = 1 \text{ bit}, \cr 
    && \Pi_{g_{n=\overline{1, N-2}}}[1] = 1 \text{ bit}, \cr 
    && \forall n, \sigma \text{ } H(X_{\sigma(1)}, X_{\sigma(2)}, \dots X_{\sigma(n)}) = \Pi_s + \sum_{i = 2}^{n-1} \Pi_{g_i}
\end{eqnarray}

\paragraph{Resolving the partial information decomposition self-contradiction}

The existence of any multivariate decomposition was previously believed to be disproved \cite{BROJ} by employing a simple example that could not be solved without discarding one of the partial information decomposition axioms. The information inside three XOR variables $O_1, O_2, O_3$ about their joint distribution $O_4 = (O_1, O_2, O_3)$ was claimed to be grouped into three $1$-bit synergistic atoms that, using our notation, correspond to $O_1\cap (O_2 \cup O_3)\cap O_4$, $O_2\cap (O_1 \cup O_3)\cap O_4$ and $O_3\cap (O_1 \cup O_2)\cap O_4$. These were summed up to give three bits of information - more than the total of two bits present in the entire system. The authors of \cite{BROJ} concluded that the non-negativity of information was not respected.

To resolve this discrepancy, first notice that partial information decomposition atoms are a subset of of the full set of atoms $\{\Pi\}$. In the system with $N$ sources of information $X_1, \dots X_N$ and target $X_{N+1}$ they lie inside the intersection $I(X_1, \dots, X_N; X_{N+1}) = \hat H((X_1\cup\dots\cup X_N)\cap X_{N+1})$ and are defined by the submatrix of the full parthood table $f_{ij}: \Xi_i \subseteq (X_1\cup X_2\cup\dots X_N)\cap X_{N+1}$. In particular, when the output is equal to the joint distribution of all inputs, the entropy of inputs coincides with mutual information and hence all atoms appear in the partial information decomposition (SI, Lemma \ref{SI:lemma:BROJresolution}). The set of atoms $\{\Pi\}$ itself is then identical to that of the system $X_1, \dots X_N$ alone with the exception of all covering numbers being increased by one to comply with the additional cover of $X_{N+1}$ (SI, Theorem \ref{SI:theorem:BROJresolution}). This is exactly the type of system that was used in \cite{BROJ}. Using the solution of XOR gate, we find
\begin{eqnarray}
    && \Pi_s[3] = 1 \text{ bit}, \cr 
    && \Pi_g[2] = 1 \text{ bit}, \cr 
    && I(O_1; O_4) = \Pi_s, \cr 
    && I(O_2; O_4) = \Pi_s, \cr 
    && I(O_3; O_4) = \Pi_s, \cr 
    && I((O_1, O_2); O_4) = \Pi_s + \Pi_g, \cr
    && I((O_1, O_3); O_4) = \Pi_s + \Pi_g, \cr
    && I((O_2, O_3); O_4) = \Pi_s + \Pi_g, \cr
    && I((O_1, O_2, O_3); O_4) = \Pi_s + \Pi_g
\end{eqnarray}
In place of three, there is only one symmetric atom $\Pi_s[3]$. The confusion in \cite{BROJ} occurred since different forms of the inclusion-exclusion principle were considered separately and it was assumed that each version would create its own synergistic atom. 

\section{Discussion}

Previous attempts for studying synergistic information using set-theoretic intuition have led to self-contradictions. In this work we point out that the non-distributivity of random variables corresponds to a well-defined variant of set-theory. We employ our results to construct a Venn-like diagram for an arbitrary $3$-variable system and demonstrate how synergism to be a direct consequence of distributivity breaking.

Our results do not fully solve the problem at hand. First, precise calculation of atom sizes was left unanswered and might require a more explicit description of information intersections. Another caveat is that although we constructed the equations that describe a self-consistent multivariate information decomposition, the existence of solution for $N$ arbitrary random variables is yet to be proven.

Nevertheless, this work lays the basis for a self-consistent multivariate theory. Our analysis reestablishes the concept of information decompositions as a foundation for further enquiry in quantifying emergence. In this context, information theory serves as a mere illustration: the  mechanism we describe offers an explanation of the nature of synergy which uses solely set-theoretic concepts and can be applied to any emergent physical system.

\backmatter

\bmhead{Acknowledgements}
We wish to thank Rotem Shalev, Prof. Amir Shpilka and Prof. Gregory Falkovich  for their insightful comments. This work has received funding from the European Research Council (ERC) under the European Union’s Horizon 2020 research and innovation program (grant agreement No. 770964), the Israeli Science Foundation grant 1727/20, and the Minerva Foundation. O.F. is the incumbent of the Henry J Leir Professorial chair.

\bibliography{sn-bibliography}

\newpage

\textbf{Supplementary information}

\begin{appendices}

\section{Properties of the extended random variable (RV) space}

To overcome the limitations of the random variable space, we extend it by adding all necessary elements to make set-theoretic operations well-defined (see Section \ref{sec2}). The extended RV space has operations of union $\cup$ and intersection $\cap$ obeying a set of axioms:
\begin{eqnarray} 
    \forall X, Y, Z
\end{eqnarray}
\begin{align*}
    X \cup \emptyset&= X & X \cap \emptyset &= \emptyset\\
    X \cup X &= X  & X \cap X &= X\\
    X \cup Y &= Y \cup X & X \cap Y &= Y \cap X\\ 
    (X\cup Y) \cup Z &= X\cup (Y \cup Z) & (X \cap Y) \cap Z &= X \cap (Y \cap Z)
\end{align*}
The relation of subset is defined as
\begin{eqnarray}\label{SI:subsetdef}
    && X\subseteq Y \Leftrightarrow X\cap Y = X \Leftrightarrow X\cup Y = Y 
\end{eqnarray}

While the postulate of distributivity is independent of the other axioms in set theory, a weaker condition ought to hold even without it:
\begin{lemma}\label{SI:lemma:subdistributivity}
    As long as the relation of subset is well-defined (\ref{SI:subsetdef}), the subdistributivity of union over intersection holds for any three variables 
    \begin{eqnarray}
        && (X \cap Z) \cup (Y \cap Z) \subseteq (X \cup Y) \cap Z
    \end{eqnarray}
\end{lemma}
\begin{proof}
    We start by showing that lemma's condition implies the following rule
    \begin{eqnarray}\label{relation1}
        && X_1\subseteq X_3, X_2\subseteq X_3 \Rightarrow (X_1\cup X_2)\subseteq X_3
    \end{eqnarray}
    Indeed, $(X_1\cup X_2)\cup X_3 = (X_1\cup X_3)\cup X_2 =$ $ X_3\cup X_2 = X_3$. 
    
    Then, for arbitrary variables $X, Y, Z$,
    \begin{eqnarray}
        && (X \cap Z) \cap ((X \cup Y) \cap Z) = (X \cup Y) \cap (Z \cap(X \cap Z)) = \cr 
        && = (X \cup Y)\cap(X \cap Z) = (X \cap Z) \Rightarrow (X \cap Z) \subseteq (X \cup Y) \cap Z,
    \end{eqnarray}
    Likewise $(Y \cap Z) \subseteq (X \cup Y) \cap Z$. By applying relation (\ref{relation1}), we get the statement of the lemma 
    \begin{eqnarray}
        && (X \cap Z) \cup (Y \cap Z) \subseteq (X \cup Y) \cap Z
    \end{eqnarray}
\end{proof}
\begin{corollary}
    The subdistributivity (or, more correctly, superdistributivity) of intersection over union also holds 
    \begin{eqnarray}
        && (X \cap Y) \cup Z \subseteq (X \cup Z) \cap (Y\cup Z)
    \end{eqnarray}
\end{corollary}

In the extended RV space the inclusion-exclusion principle for $2$ variables, as expected, remains unaffected by the lack of distributivity
\begin{lemma}\label{SI:lemma:inclusionexclusion2}
    The size of the union of two extended RV space members is related to their own sizes and the size of their intersection as
    \begin{eqnarray}\label{SI:incexc2}
        && \hat H(X\cup Y) =  \hat H(X) +  \hat H(Y) -  \hat H( X\cap  Y)
    \end{eqnarray}
\end{lemma}
\begin{proof}
    We will rewrite the left side $H(X\cup Y)$ as a union of two disjoint pieces
    \begin{eqnarray}
        && X\cup (Y\backslash X) = X\cup (X\cap Y)\cup (Y\backslash X) = X\cup Y
    \end{eqnarray}
    At the same time, by definition $X\cap (Y\backslash X) = \emptyset$. We may use the additivity of the measure (\ref{sec2:measure}) to write the size of union as a sum of sizes of its disjoint parts
    \begin{eqnarray}
         && \hat H(X\cup Y) =  \hat H(X) +  \hat H(Y\backslash X)
    \end{eqnarray}
    Repeating the same steps in order to decompose the second summand into two more terms concludes the proof.
    \begin{eqnarray}
        && (Y\backslash  X)\cap (X\cap  Y) = ((Y\backslash  X)\cap X)\cap  Y = \emptyset, \cr 
        && (Y\backslash X)\cup (X\cap Y) = Y, \cr 
        &&  \hat H(Y) = \hat H(Y\backslash X) + \hat H(X \cap Y)
    \end{eqnarray}
\end{proof}

The inclusion-exclusion principle for $3$ variables along with all the terms from the set-theoretic version, contains a peculiar extra term related to the failure of distributivity
\begin{theorem}\label{SI:theorem:incexc3}
    The size of triple union is related to the sizes of individual terms, their intersections and the distributivity-breaking difference $\Delta \hat H$
    \begin{eqnarray}
        &&  \hat H( X_1\cup  X_2\cup  X_3) = \cr 
        && = \hat H( X_1) +  \hat H( X_2) +  \hat H( X_3) - \cr 
        && - \hat H( X_1\cap  X_2) -  \hat H( X_1\cap  X_3) -  \hat H( X_2\cap  X_3) + \cr 
        && + \hat H( X_1\cap  X_2\cap  X_3) - \Delta \hat H,
    \end{eqnarray}  
    where the last term is invariant with respect to permutations of indices $\sigma$ and equal to 
    \begin{eqnarray}
        && \Delta \hat H = \hat H(((X_{\sigma(1)}\cup X_{\sigma(2)})\cap X_{\sigma(3)}) \backslash ((X_{\sigma(1)}\cap X_{\sigma(3)})\cup(X_{\sigma(2)}\cap X_{\sigma(3)}))) 
    \end{eqnarray}
\end{theorem}
\begin{proof}
    We begin by choosing two of three variables (or extended RV space members) on the left side and grouping them in order to use the result of the previous lemma (\ref{SI:incexc2}).
    \begin{eqnarray}\label{SI1}
        &&  \hat H(( X_1\cup  X_2)\cup  X_3) = \hat H( X_1\cup  X_2) +  \hat H( X_3) - \hat H(( X_1\cup  X_2)\cap  X_3)
    \end{eqnarray} 
    The first term is easily decomposed further using Lemma \ref{SI:lemma:inclusionexclusion2}. In order to proceed with the third term we define the distributivity-breaking difference as 
    \begin{eqnarray}
        && \Delta X_{123} = \big(( X_1\cup X_2)\cap  X_3\big)\backslash\big(( X_1\cap X_3)\cup( X_2\cap X_3)\big),
    \end{eqnarray}
    and apply the second axiom of the measure 
    \begin{eqnarray}\label{SI2}
        &&  \hat H(( X_1\cup  X_2)\cap  X_3) = \hat H((( X_1\cap  X_3)\cup ( X_2\cap  X_3))\cup \Delta X_{123}) = \cr 
        && = \hat H(( X_1\cap  X_3)\cup ( X_2\cap  X_3)) +  \hat H(\Delta X_{123}) 
    \end{eqnarray}
    Applying (\ref{SI:incexc2}) once again,
    \begin{eqnarray}
        &&  \hat H(( X_1\cap  X_3)\cup ( X_2\cap  X_3)) = \cr 
        && = \hat H( X_1\cap  X_3) +  \hat H( X_2\cap  X_3) - \hat H( X_1\cap  X_2\cap  X_3)
    \end{eqnarray}
    and combining everything into the final form
    \begin{eqnarray}
        &&  \hat H( X_1\cup  X_2\cup  X_3) = \cr 
        && = \hat H( X_1) +  \hat H( X_2) +  \hat H( X_3) - \cr 
        && -  \hat H( X_1\cap  X_2) -  \hat H( X_1\cap  X_3) -  \hat H( X_2\cap  X_3) + \cr 
        && + \hat H( X_1\cap  X_2\cap  X_3) -  \hat H(\Delta X_{123})
    \end{eqnarray}  
    The only term that depends on the order of putting brackets in (\ref{SI1}) is $\hat H(\Delta X_{123})$. Due to the associativity and commutativity of both union and intersection, we conclude that it is the only part of the equation that is not symmetric with respect to permutations of indices 
    \begin{eqnarray}
        && \Delta X_{123}\neq \Delta X_{132}\neq \Delta X_{231},
    \end{eqnarray}
    Its size is therefore bound to be the same in all three cases. Defining a single function equal to this value concludes the proof
    \begin{eqnarray}
        && \Delta \hat H = \hat H(\Delta X_{123})
    \end{eqnarray}
\end{proof}

The operation of taking the difference $\backslash$ in extended RV space can have more than one outcome. It can be shown already on the XOR gate example. Taking the variable that represents the whole system $W=O_1 \cup O_2 \cup O_3$, we have two candidates $O_2, O_3$ for the result of difference $W\backslash O_1$. Substituting them into the definition (\ref{sec2:differencedef}), we find out that both are valid, despite being explicitly unequal
\begin{eqnarray}\label{differenceexample}
    && O_{i=2, 3} \cap O_1 = \emptyset, \cr 
    && O_i \cup (W\cap O_1) = O_i \cup O_1 = W
\end{eqnarray}

\section{Information atoms}

A convenient notation of antichains was proposed in partial information decomposition \cite{WilliamsAndBeer,BROJ} to describe pieces of information. Let us denote each joint distribution by the collection of variables' indices
\begin{eqnarray}\label{SI:originalnodes}
    && (X_{i_1}, X_{i_2}, \dots, X_{i_m}) \rightarrow \{i_1 i_2 \dots i_m\} = \mathbf{A}
\end{eqnarray}
There is a trivial partial order $\mathbf{A} \preceq \mathbf{B} \Leftrightarrow \forall i\in \mathbf{A} \text{ } i\in \mathbf{B}$ and we can use it to represent the intersections. A set of strong antichains $\alpha \in \mathcal{A}(N)$ is taken on the above poset
\begin{eqnarray}
    && \alpha = \mathbf{A}_1\mathbf{A}_2\dots \mathbf{A}_n = \{i_{11}\dots i_{1m_1}\}\{i_{21}\dots i_{2m_2}\}\dots\{i_{n1}\dots i_{nm_n}\},
\end{eqnarray}
where all indices are chosen from $\overline{1, N}$ and never coincide $i_{ab}\neq i_{cd}$. The partial order $\preceq$ can be extended to antichains 
\begin{eqnarray}\label{SI:partialorder2}
    && \alpha \preceq \beta \Leftrightarrow \forall \mathbf{B} \in \beta \text{ } \exists \mathbf{A} \in \alpha: \mathbf{A} \preceq \mathbf{B}
\end{eqnarray}

Now a general inclusion-exclusion term (\ref{sec3:term}) in an $N$-variable system can be denoted by an antichain $\alpha \in \mathcal{A}(N)$ 
\begin{eqnarray}
    && \Xi_\alpha[C = n] = \bigcap_{j=1}^n \left(\bigcup_{k=1}^{m_j} X_{i_{jk}}\right)
\end{eqnarray}
The covering $C$ is always equal to the cardinality of the corresponding antichain (the number of brackets $\{\}$)
\begin{eqnarray}
    && C = n = |\alpha|
\end{eqnarray}
The inclusion order on $\Xi$-terms follows from the antichain order $\preceq$. Note that the latter is independent of the chosen random variables and holds for every system
\begin{eqnarray}
    && \alpha \preceq \beta \Rightarrow \Xi_\alpha \subseteq \Xi_\beta
\end{eqnarray}

This allows us to replace the first index of the parthood table $f_{ij}$ with an antichain and simplify the formulation of multivariate theory's axioms 
\begin{eqnarray}\label{axioms}
    && \hat H(\Xi_{\alpha}) = \sum_i f_{\alpha i}\Pi_i[c_i], \cr
    && \Xi_\alpha \subseteq \Xi_\beta \Rightarrow\forall i\text{ } f_{\alpha i} \leq f_{\beta i}, \cr 
    && c_i = \max_{\alpha: f_{\alpha i} = 1} |\alpha| \cr 
    && \sum_{k=1}^N H(X_k) = \sum_i c_i \Pi_i[c_i]
\end{eqnarray}

\begin{lemma}\label{SI:lemma:equalsizefunctions}
    Two inclusion-exclusion terms that are equal as members of extended RV space, have identical parthood matrix rows 
    \begin{eqnarray}
        && \Xi_\alpha = \Xi_\beta \Rightarrow \forall i: \Pi_i > 0 \text{ } f_{\alpha i} = f_{\beta i}
    \end{eqnarray}
\end{lemma}

\section{Examples}

\paragraph{Set-theoretic solution}

This is a complete replica of set theory, fully compliant with the distributivity axiom. For $N=3$ variables, condition (\ref{sec2:settheoreticsolutioncond}) is necessary and sufficient for non-negativity of all atoms
\begin{eqnarray}\label{SI:settheor3derivation}
    && \Pi_{\{1\}\{2\}\{3\}} = I(X_1; X_3) + I(X_2; X_3) - I(X_1, X_2; X_3) \geq 0, \cr 
    && \Pi_{\{i\}\{j\}} = I(X_i; X_j) - \Pi_{\{1\}\{2\}\{3\}} = I(X_i; X_j| X_{k}) \geq 0, \cr 
    && \Pi_{\{i\}} = - H(X_{j}, X_{k}) + H(X_1, X_2, X_3) \geq 0
\end{eqnarray}

For an arbitrary number of variables $N$ there is no variability in the inclusion-exclusion formulas and the atoms are recovered via the Möbius inversion with respect to the antichain order $\preceq$. Let us also denote the atoms by a special subset of antichains $\iota$ with a single index in each bracket
\begin{eqnarray}\label{SI:settheory1}
    && \hat H(\Xi_\alpha) = \sum_{\iota \preceq \alpha} \Pi_{\iota}[n], \qquad  \iota = \{i_1\}\{i_2\}\dots \{i_n\}, \cr
    && \Pi_{\iota} = \sum_{m=n}^{N} (-1)^{m-n}\sum_{i_{n+1}, \dots i_m} I_m(X_{i_1}; \dots X_{i_m}),
\end{eqnarray}
where $I_m$ is the $m^{\text{th}}$ order interaction information function
\begin{eqnarray}\label{SI:interactioninfo}
    && I_m(X_1; \dots X_m) = \sum_{k=1}^m (-1)^{k-1}\sum_{j_1, \dots j_k} H(X_{j_1}, X_{j_2}, \dots, X_{j_k})
\end{eqnarray}
Since all atoms must be non-negative, set-theoretic solution is only applicable when
\begin{eqnarray}\label{settheoreticatomscond}
    && \forall i_1\neq i_2\neq \dots \neq i_n \sum_{m=n}^{N} (-1)^{m-n}\sum_{i_{n+1}, \dots i_m} I_m(X_{i_1}; \dots X_{i_m}) \geq 0 
\end{eqnarray}

Set-theoretic systems never exhibit synergistic properties. The following result can be understood in the sense that the lack of distributivity is a necessary condition for the existence of synergy in any $N$-variable system 
\begin{lemma}\label{SI:lemma:nosynergyinsettheory}
    In set-theoretic system mutual information is always subadditive
    \begin{eqnarray}
        && I(X_1; X_{N+1}) + \dots + I(X_N; X_{N+1}) \geq I((X_1, \dots X_N); X_{N+1})
    \end{eqnarray}
\end{lemma}
\begin{proof}
    Let us substitute the atoms (\ref{SI:settheory1}) into the inequality. 
    \begin{eqnarray}\label{SI:sum}
        && \sum_{k=1}^N\sum_{\iota' \preceq \{k\}\{N+1\}}\Pi_{\iota'} \geq \sum_{\iota \preceq \{12\dots N\}\{N+1\}}\Pi_{\iota}
    \end{eqnarray}
    For any atom on the right side, we have by (\ref{SI:partialorder2})
    \begin{eqnarray}
        && \iota \preceq \{12\dots N\}\{N+1\} \Rightarrow \cr 
        && \Rightarrow \exists \{i_a\}, \{i_b\} \in \iota:  \begin{cases} \{i_a\}\preceq \{12\dots N\} \\ \{i_b\}\preceq \{N+1\} \end{cases}
    \end{eqnarray}
    Since $\iota$ is composed of single indices, we have 
    \begin{eqnarray}
        && i_a = \overline{1, N}, \cr 
        && i_b = N+1
    \end{eqnarray}
    Then such term can also be found on the left side of (\ref{SI:sum})
    \begin{eqnarray}
        && \exists \iota' \preceq \{i_a\}\{N+1\}: \iota = \iota' 
    \end{eqnarray}
    The non-negativity of all atoms concludes the proof.
\end{proof}

\paragraph{XOR gate}

The XOR gate contains a completely different set of atoms. With three mutually independent initial variables, the set of inclusion-exclusion terms simplifies to  
\begin{align*}
    \forall i\neq j\neq k & & & \\
    \Xi_{\{123\}}[1] &=  O_1\cup O_2\cup O_3 &  \hat H(\Xi_{\{123\}}) &= 2, \\
    \Xi_{\{ij\}}[1] &=  O_i\cup O_j &  \hat H(\Xi_{\{ij\}}) &= 2, \\
    \Xi_{\{i\}}[1] &=  O_i &  \hat H(\Xi_{\{i\}}) &= 1, \\ 
    \Xi_{\{ij\}\{k\}}[2] &= ( O_i\cup O_j)\cap O_k &  \hat H(\Xi_{\{ij\}\{k\}}) &= 1
\end{align*}

In the extended RV space $\Xi_{\{ij\}\{k\}} = \Xi_{\{k\}}, \Xi_{\{ij\}} = \Xi_{\{123\}}$ and by Lemma \ref{SI:lemma:equalsizefunctions} we only need to find decompositions of $\Xi_{\{i\}}$ and $\Xi_{\{123\}}$. Due to the symmetry of the problem, decomposition of $\Xi_{\{i\}}$ may contain three types of atoms: $3$ distinct atoms $\Pi_{i}$, each being a part of only the respective $\Xi_{i}$; $3$ distinct atoms $\Pi_{i, j}$, each being a part of both specified terms; or one symmetrically shared $\Pi_s$ as we have guessed in (\ref{sec1:XORsatomdef})
\begin{eqnarray}\label{derivation31}
    && \hat H(\Xi_{\{ij\}\{k\}}[2]) =  \hat H(\Xi_{\{k\}}[1]) = \cr 
    && = \Pi_{k}[2] + \Pi_{i, k}[2] + \Pi_{j, k}[2] + \Pi_s[2]
\end{eqnarray}
The coverings are calculated by definition (\ref{axioms}). For $\Xi_{\{123\}}$ one more atom $\Pi_g$ may be added
\begin{eqnarray}\label{derivation32}
    &&  \hat H(\Xi_{\{ij\}}[1]) =  \hat H(\Xi_{\{123\}}[1]) = \Pi_{1}[2] + \Pi_{2}[2] + \Pi_{3}[2] + \cr 
    && + \Pi_{1, 2}[2] + \Pi_{2, 3}[2] + \Pi_{1, 3}[2] + \Pi_s[2] + \Pi_g[1]
\end{eqnarray}
The following parthood table contains columns for all atoms discussed above. 
\begin{center}
\begin{tabular}{|c|c|c|c|c|c|c|c|c|}
     \hline
       $f$ & $\Pi_s$ & $\Pi_g$ & $\Pi_{1}$ & $\Pi_{2}$ & $\Pi_{3}$ & $\Pi_{1, 2}$ & $\Pi_{1, 3}$ & $\Pi_{2, 3}$ \\ [0.5ex] 
     \hline
     $\{1\}\{2\}\{3\}$ & 0 & 0 & 0 & 0 & 0 & 0 & 0 & 0 \\ 
     \hline
     $\{1\}\{2\}$ & 0 & 0 & 0 & 0 & 0 & 0 & 0 & 0 \\
     \hline
     $\{1\}\{3\}$ & 0 & 0 & 0 & 0 & 0 & 0 & 0 & 0 \\ 
     \hline
     $\{2\}\{3\}$ & 0 & 0 & 0 & 0 & 0 & 0 & 0 & 0 \\ 
     \hline
     $\{12\}\{3\}$ & 1 & 0 & 0 & 0 & 1 & 0 & 1 & 1 \\ 
     \hline
     $\{13\}\{2\}$ & 1 & 0 & 0 & 1 & 0 & 1 & 0 & 1 \\ 
     \hline
     $\{23\}\{1\}$ & 1 & 0 & 1 & 0 & 0 & 1 & 1 & 0 \\ 
     \hline
     $\{1\}$ & 1 & 0 & 1 & 0 & 0 & 1 & 1 & 0 \\ 
     \hline
     $\{2\}$ & 1 & 0 & 0 & 1 & 0 & 1 & 0 & 1 \\ 
     \hline
     $\{3\}$ & 1 & 0 & 0 & 0 & 1 & 0 & 1 & 1 \\ 
     \hline
     $\{12\}$ & 1 & 1 & 1 & 1 & 1 & 1 & 1 & 1 \\ 
     \hline
     $\{13\}$ & 1 & 1 & 1 & 1 & 1 & 1 & 1 & 1 \\ 
     \hline
     $\{23\}$ & 1 & 1 & 1 & 1 & 1 & 1 & 1 & 1 \\ 
     \hline
     $\{123\}$ & 1 & 1 & 1 & 1 & 1 & 1 & 1 & 1 \\ 
     \hline
    \end{tabular}
\end{center}
In a symmetric solution, the atom sizes are invariant with respect to index permutations, hence let
\begin{eqnarray}
    && \Pi_s = x, \cr 
    && \Pi_{i, k} = y, \cr 
    && \Pi_{i} = 1-2y-x, \cr 
    && \Pi_g = 2 - x - 3y - 3(1-2y-x) = \cr 
    && = 2x + 3y - 1
\end{eqnarray}
Substituting this into the information conservation law,
\begin{eqnarray}
    && \sum_i H(X_i) = 3 = 2x + 2*3y + 2*3*(1-2y-x) +  \cr 
    && + 2x + 3y - 1 = 5 - 2x - 3y, \cr 
    && 2x + 3y = 2
\end{eqnarray}
But we know that all atoms have non-negative sizes, which means that
\begin{eqnarray}
    && \Pi_{i} = 1-2y-x = -0.5y \geq 0, \cr 
    && y = 0, \cr 
    && x = 1
\end{eqnarray}
\begin{theorem}\label{SI:theorem:XORunique}
    The XOR gate has a unique symmetric decomposition
    \begin{eqnarray}
        && \Pi_s[2] = 1 \text{ bit}, \cr
        && \Pi_g[1] = 1 \text{ bit}
    \end{eqnarray}
\end{theorem}
\begin{center}
\begin{tabular}{|c|c|c|}
     \hline
       $f$ & $\Pi_s$ & $\Pi_g$ \\ [0.5ex] 
     \hline
     $\{1\}\{2\}\{3\}$ & 0 & 0 \\ 
     \hline
     $\{1\}\{2\}$ & 0 & 0 \\
     \hline
     $\{1\}\{3\}$ & 0 & 0 \\ 
     \hline
     $\{2\}\{3\}$ & 0 & 0 \\ 
     \hline
     $\{12\}\{3\}$ & 1 & 0 \\ 
     \hline
     $\{13\}\{2\}$ & 1 & 0 \\ 
     \hline
     $\{23\}\{1\}$ & 1 & 0 \\ 
     \hline
     $\{1\}$ & 1 & 0 \\ 
     \hline
     $\{2\}$ & 1 & 0 \\ 
     \hline
     $\{3\}$ & 1 & 0 \\ 
     \hline
     $\{12\}$ & 1 & 1 \\ 
     \hline
     $\{13\}$ & 1 & 1 \\ 
     \hline
     $\{23\}$ & 1 & 1 \\ 
     \hline
     $\{123\}$ & 1 & 1 \\ 
     \hline
    \end{tabular}
\end{center}

\paragraph{N-parity}

A generalization of XOR gate is the $N$-parity setup, also a symmetric system, for which one of the variables is fully determined by the combination of all the others
\begin{eqnarray}\label{Nparity}
    && X_{\overline{1, N}} = \begin{cases} 0, & 50\% \\ 1, & 50\% \end{cases}, \cr 
    && \forall i=\overline{1, N}\text{ } X_i \equiv \sum_{j\neq i} X_j \mod 2
\end{eqnarray}
The inclusion-exclusion terms coinciding with the entropies are
\begin{eqnarray}\label{term1}
    && \Xi_{\{i_1i_2\dots i_n\}}[1] = \bigcup_{k=\overline{1, n}} X_{i_k}, \cr 
    && \hat H(\Xi_{\{i_1i_2\dots i_n\}}) = \min(n, N-1)
\end{eqnarray}
From this one can see that dividing all $N$ variables into two sets and intersecting corresponding joint distributions gives us a $1$-bit $2$-covered term.
\begin{eqnarray}\label{term2}
    && \Xi_{\{i_1i_2\dots i_n\}\{i_{n+1}\dots i_N\}}[1] = \left(\bigcup_{k=\overline{1, n}} X_{i_k}\right)\cap\left(\bigcup_{l=\overline{n+1, N}} X_{i_l}\right), \cr 
    &&  \hat H(\Xi_{\{i_1i_2\dots i_n\}\{i_{n+1}\dots i_N\}}) = 1
\end{eqnarray}
The rest of inclusion-exclusion terms are empty. A solution can be easily guessed: a single symmetric $2$-covered atom $\Pi_s[2]=1$ and a set of $N-2$ ghost atoms $\Pi_{g_k}[1]=1, k=\overline{1, N-2}$, such that
\begin{eqnarray}
    && \hat H(\Xi_{\{i_1i_2\dots i_n\}\{i_{n+1}\dots i_N\}}) = \Pi_s, \cr 
    && \hat H(\Xi_{\{i_1i_2\dots i_n\}}) = \Pi_s + \sum_{k=1}^{n-2}\Pi_{g_k}
\end{eqnarray}
We immediately see that the information conservation law is satisfied
\begin{eqnarray}\label{Nparity2}
    && \sum_i H(X_i) = N = 2\Pi_s + \sum_{k=1}^{N-2} \Pi_{g_k}
\end{eqnarray}

\paragraph{Arbitrary trivariate system}

\begin{eqnarray}\label{3PHDequations}
    && H(X_1) = \Pi_{\{1\}\{2\}\{3\}}[3] + \Pi_s[2] + \Pi_{\{1\}\{2\}}[2] + \Pi_{\{1\}\{3\}}[2] + \Pi_{\{1\}}[1], \cr 
    && H(X_2) = \Pi_{\{1\}\{2\}\{3\}}[3] + \Pi_s[2] + \Pi_{\{1\}\{2\}}[2] + \Pi_{\{2\}\{3\}}[2] + \Pi_{\{2\}}[1], \cr 
    && H(X_3) = \Pi_{\{1\}\{2\}\{3\}}[3] + \Pi_s[2] + \Pi_{\{1\}\{3\}}[2] + \Pi_{\{2\}\{3\}}[2] + \Pi_{\{3\}}[1], \cr 
    && H(X_1, X_2) = \Pi_{\{1\}\{2\}\{3\}}[3] + \Pi_s[2] + \Pi_{\{1\}\{2\}}[2] + \Pi_{\{1\}\{3\}}[2] + \Pi_{\{2\}\{3\}}[2] + \cr 
    && + \Pi_{\{1\}}[1] + \Pi_{\{2\}}[1] + \Pi_g[1], \cr 
    && H(X_1, X_3) = \Pi_{\{1\}\{2\}\{3\}}[3] + \Pi_s[2] + \Pi_{\{1\}\{2\}}[2] + \Pi_{\{1\}\{3\}}[2] + \Pi_{\{2\}\{3\}}[2] + \cr 
    && + \Pi_{\{1\}}[1] + \Pi_{\{3\}}[1] + \Pi_g[1], \cr 
    && H(X_2, X_3) = \Pi_{\{1\}\{2\}\{3\}}[3] + \Pi_s[2] + \Pi_{\{1\}\{2\}}[2] + \Pi_{\{1\}\{3\}}[2] + \Pi_{\{2\}\{3\}}[2] + \cr 
    && + \Pi_{\{2\}}[1] + \Pi_{\{3\}}[1] + \Pi_g[1], \cr 
    && H(X_1, X_2, X_3) = \Pi_{\{1\}\{2\}\{3\}}[3] + \Pi_s[2] + \Pi_{\{1\}\{2\}}[2] + \Pi_{\{1\}\{3\}}[2] + \Pi_{\{2\}\{3\}}[2] + \cr 
    && + \Pi_{\{1\}}[1] + \Pi_{\{2\}}[1] + \Pi_{\{3\}}[1] + \Pi_g[1], \cr 
    && \Pi_s[2] = \Pi_g[1]
\end{eqnarray}

\begin{center}
\resizebox{\columnwidth}{!}{
\begin{tabular}{|c|c|c|c|c|c|c|c|c|c|}
     \hline
       $f$ & $\Pi_{\{1\}\{2\}\{3\}}$ & $\Pi_s$ & $\Pi_{\{1\}\{2\}}$ & $\Pi_{\{1\}\{3\}}$ & $\Pi_{\{2\}\{3\}}$ & $\Pi_{\{1\}}$ & $\Pi_{\{2\}}$ & $\Pi_{\{3\}}$ & $\Pi_g$ \\ [0.5ex] 
     \hline
     $\{1\}\{2\}\{3\}$ & 1 & 0 & 0 & 0 & 0 & 0 & 0 & 0 & 0 \\ 
     \hline
     $\{1\}\{2\}$ & 1 & 0 & 1 & 0 & 0 & 0 & 0 & 0 & 0 \\
     \hline
     $\{1\}\{3\}$ & 1 & 0 & 0 & 1 & 0 & 0 & 0 & 0 & 0 \\ 
     \hline
     $\{2\}\{3\}$ & 1 & 0 & 0 & 0 & 1 & 0 & 0 & 0 & 0 \\ 
     \hline
     $\{12\}\{3\}$ & 1 & 1 & 0 & 1 & 1 & 0 & 0 & 0 & 0 \\ 
     \hline
     $\{13\}\{2\}$ & 1 & 1 & 1 & 0 & 1 & 0 & 0 & 0 & 0 \\ 
     \hline
     $\{23\}\{1\}$ & 1 & 1 & 1 & 1 & 0 & 0 & 0 & 0 & 0 \\ 
     \hline
     $\{1\}$ & 1 & 1 & 1 & 1 & 0 & 1 & 0 & 0 & 0 \\ 
     \hline
     $\{2\}$ & 1 & 1 & 1 & 0 & 1 & 0 & 1 & 0 & 0 \\ 
     \hline
     $\{3\}$ & 1 & 1 & 0 & 1 & 1 & 0 & 0 & 1 & 0 \\ 
     \hline
     $\{12\}$ & 1 & 1 & 1 & 1 & 1 & 1 & 1 & 0 & 1 \\ 
     \hline
     $\{13\}$ & 1 & 1 & 1 & 1 & 1 & 1 & 0 & 1 & 1 \\ 
     \hline
     $\{23\}$ & 1 & 1 & 1 & 1 & 1 & 0 & 1 & 1 & 1 \\ 
     \hline
     $\{123\}$ & 1 & 1 & 1 & 1 & 1 & 1 & 1 & 1 & 1 \\ 
     \hline
    \end{tabular}
    }
\end{center}

\begin{lemma}\label{SI:lemma:arbitrary3PHD}
    Any system of three random variable can be decomposed into a set of non-negative atoms (\ref{3PHDequations}).
\end{lemma}
\begin{proof}
    One can find the sizes of atoms $\Pi_{\{i\}}$ from the last $4$ equations in the system
    \begin{eqnarray}
        && \Pi_{\{1\}} = H(X_1, X_2, X_3) - H(X_2, X_3) \geq 0, \cr 
        && \Pi_{\{2\}} = H(X_1, X_2, X_3) - H(X_1, X_3) \geq 0, \cr
        && \Pi_{\{3\}} = H(X_1, X_2, X_3) - H(X_1, X_2) \geq 0
    \end{eqnarray}

    For the rest of set-theoretic atoms we have
    \begin{eqnarray}
        && I(X_1; X_2) = \Pi_{\{1\}\{2\}\{3\}} + \Pi_{\{1\}\{2\}}, \cr 
        && I(X_1; X_3) = \Pi_{\{1\}\{2\}\{3\}} + \Pi_{\{1\}\{3\}}, \cr 
        && I(X_2; X_3) = \Pi_{\{1\}\{2\}\{3\}} + \Pi_{\{2\}\{3\}}
    \end{eqnarray}
    To satisfy the non-negativity requirement, we need
    \begin{eqnarray}
        && \Pi_{\{1\}\{2\}} = I(X_1; X_2) - \Pi_{\{1\}\{2\}\{3\}} \geq 0, \cr
        && \Pi_{\{1\}\{3\}} = I(X_1; X_3) - \Pi_{\{1\}\{2\}\{3\}} \geq 0, \cr
        && \Pi_{\{2\}\{3\}} = I(X_2; X_3) - \Pi_{\{1\}\{2\}\{3\}} \geq 0,
    \end{eqnarray}
    which is equivalent to 
    \begin{eqnarray}
        && 0 \leq \Pi_{\{1\}\{2\}\{3\}} \leq \min\left(I(X_1; X_2), I(X_1; X_3), I(X_2; X_3)\right) 
    \end{eqnarray}
    The last independent equation can be used in the following form (see \ref{SI:interactioninfo}) 
    \begin{eqnarray}
        && I_3(X_1; X_2; X_3) = \Pi_{\{1\}\{2\}\{3\}} - \Pi_s,
    \end{eqnarray}
    therefore 
    \begin{eqnarray}
        && \Pi_s = \Pi_g = \Pi_{\{1\}\{2\}\{3\}} - I_3(X_1; X_2; X_3) \geq 0
    \end{eqnarray}
    The obtained set of conditions is indeed self-consistent, as for any three random variables
    \begin{eqnarray}
        && \min\left(I(X_1; X_2), I(X_1; X_3), I(X_2; X_3)\right) \geq I_3(X_1; X_2; X_3)
    \end{eqnarray}
\end{proof}

\section{Partial information decomposition (PID)}

The partial information decomposition atoms are only a subset of all atoms $\Pi$. Yet, for some systems it may be equal to the full set. Indeed, when the output is exactly the joint distribution of all inputs, it essentially 'covers' the whole diagram of the system of inputs. The entropies of inputs completely turn into mutual information about the output.

\begin{lemma}\label{SI:lemma:BROJresolution}
    The partial information decomposition with inputs $X_1, \dots X_N$ and their joint distribution chosen as an output $X_{N+1} =(X_1, \dots, X_N)$ contains all information atoms $\Pi$ of the system $X_1, \dots X_{N+1}$.
\end{lemma}
\begin{proof}
     The PID atoms are by definition the ones contained in the intersection of the form 
    \begin{eqnarray}
        && (X_1\cup X_2\cup\dots\cup X_N)\cap X_{N+1}
    \end{eqnarray}
    By conditions of the lemma, in extended random variable space we have 
    \begin{eqnarray}
        && (X_1\cup X_2\cup\dots\cup X_N)\cap X_{N+1} = \cr  
        && = (X_1\cup X_2\cup\dots\cup X_N) = \cr 
        && = (X_1\cup X_2\cup\dots\cup X_N)\cup X_{N+1}
    \end{eqnarray}
    Applying Lemma \ref{SI:lemma:equalsizefunctions} concludes the proof.
\end{proof}

A stronger statement can be made that the whole structure of the resulting $N+1$ variable decomposition is equivalent to the lesser decomposition of just the inputs $X_1, \dots X_N$ with a single extra covering added to each atom to account for the output $X_{N+1}$ 'covering' the whole system one more time.

\begin{theorem}\label{SI:theorem:BROJresolution}
    A decomposition for the $N+1$ variable system $X_1, \dots X_{N+1}$ with
    \begin{eqnarray}
        && X_{N+1} = X_1\cup X_2\cup\dots\cup X_N
    \end{eqnarray}
    defined by a set of atoms $\{\Pi_i[c_i]\}_{i\in I}$ and parthood table $f$ can be obtained from the decomposition $\{\tilde \Pi_j[\tilde c_j]\}_{j\in J}, \tilde f$ of the $N$ variable system $X_1, \dots X_N$ as
    \begin{eqnarray}\label{temp0}
        && \forall \alpha\in \mathcal{A}(N+1), i\in I \text{ } \begin{cases}
            f_{\alpha i} = \tilde f_{F(\alpha) j(i)} \\
            \Pi_i = \tilde \Pi_{j(i)} \\
            c_i = \tilde c_{j(i)} + 1
        \end{cases}
    \end{eqnarray}
    where $j(i)$ is a bijection of indices and (surjective) function $F: \mathcal{A}(N+1)\rightarrow \mathcal{A}(N)$ removes any bracket containing index $N+1$ from the antichain.
\end{theorem}
\begin{proof}
    Examining the inclusion-exclusion terms, we find that
    \begin{eqnarray}\label{temp}
        && \forall \alpha\in \mathcal{A}(N+1) \text{ } \Xi_{\alpha} = \Xi_{F(\alpha)}
    \end{eqnarray}
    By Lemma \ref{SI:lemma:equalsizefunctions} this guarantees the equivalence of the corresponding parthood table rows
    \begin{eqnarray}\label{temp1}
        && \forall i\in I, \alpha\in \mathcal{A}(N+1) \text{ } f_{\alpha i} = f_{F(\alpha) i}
    \end{eqnarray}
    Now we need to determine the parthood table rows only for $\alpha \in \text{Im}(F) = \mathcal{A}(N)$. Knowing the solution $\{\tilde \Pi\}_{j\in J}$ for the $N$-variable system, we substitute the same atoms into the larger $N+1$ variable system and define a bijection of indices $j(i)$
    \begin{eqnarray} 
        && \hat H(\Xi_{\alpha}) = \hat H(\Xi_{F(\alpha)}) = \sum_{j\in J} \tilde f_{F(\alpha) j}\tilde \Pi_j = \sum_{i\in I} f_{\alpha i}\Pi_i, \cr
        && \forall i\in I, \alpha\in \mathcal{A}(N+1) \qquad \Pi_i = \tilde \Pi_{j(i)}, f_{\alpha i} = \tilde f_{F(\alpha) j(i)}
    \end{eqnarray}
    Finally, to ensure the validity of the new solution, we check its compliance with axioms (\ref{axioms})
    \begin{enumerate}
        \item Monotonicity
        \begin{eqnarray}
            && \Xi_\alpha \subseteq \Xi_\beta \Rightarrow \Xi_{F(\alpha)} \subseteq \Xi_{F(\beta)} \Rightarrow \forall i\in I\text{ } f_{\alpha i} = \tilde f_{F(\alpha) j(i)} \leq \tilde f_{F(\beta) j(i)} = f_{\beta i}
    \end{eqnarray}
        \item Covering numbers
        \begin{eqnarray}
            && \begin{cases} N+1 \in \alpha \Rightarrow |\alpha| = |F(\alpha)|+1 \\ N+1 \notin \alpha \Rightarrow |\alpha| = |F(\alpha)| \end{cases}
        \end{eqnarray}
        \begin{eqnarray}
            && c_i = \max_{\alpha\in \mathcal{A}(N+1): f_{\alpha i}=1} |\alpha| = \max_{\beta\in \mathcal{A}(N): \tilde f_{\beta j(i)}=1} |\beta| + 1 = \tilde c_{j(i)} + 1
        \end{eqnarray}
        \item Information conservation law
        \begin{eqnarray}
            && H(X_{N+1}) = H(X_1, \dots X_N) = \sum_{j\in J} \tilde \Pi_j, \qquad \sum_{k=1}^{N} H(X_k) = \sum_{j\in J} \tilde c_j \tilde \Pi_j, \cr
            && \sum_{k=1}^{N+1} H(X_k) = \sum_{j\in J} \tilde c_i\tilde \Pi_i + H(X_{N+1}) = \sum_{j\in J} (\tilde c_j+1)\tilde \Pi_j = \sum_{i\in I} c_i\Pi_i
        \end{eqnarray}
    \end{enumerate}
 
\end{proof}

\end{appendices}

\end{document}